\documentclass[letterpaper]{article}
\usepackage{aaai24}
\usepackage{times}
\usepackage{helvet}
\usepackage{tikz}
\usepackage{courier}
\usepackage[hyphens]{url}
\usepackage{graphicx}
\usepackage{natbib}
\usepackage{caption}
\usepackage{algorithm}
\usepackage{algorithmic}
\usepackage{todonotes}
\usepackage{dsfont}
\usepackage{amsmath}
\usepackage{amsthm}
\usepackage{enumitem}
\usepackage{xspace}
\usepackage{booktabs}
\usepackage{multirow}
\usepackage{anyfontsize}
\usepackage{colortbl}
\usepackage{xcolor}

\newtheorem{definition}{Definition}
\newtheorem{proposition}{Proposition}

\newcommand{\bsfigurewidth}[4][tb]{%
  \begin{figure}[#1]
    \centering
    \includegraphics[width=#2\columnwidth]{#3}
    \vglue 0ex plus 0.5ex minus 0.5ex
    {\caption{\small #4.}\label{#3}}%
  \end{figure}}

\usepackage{newfloat}

\nocopyright
\title{Security Decisions for Cyber-Physical Systems based on Solving Critical Node Problems with Vulnerable Nodes}
\author {
    Jens Otto\textsuperscript{\rm 1},
    Niels Grüttemeier\textsuperscript{\rm 1},
    Felix Specht\textsuperscript{\rm 1}
}
\affiliations {
    \textsuperscript{\rm 1}Fraunhofer IOSB-INA, Lemgo, Germany\\
    jens.otto@iosb-ina.fraunhofer.de, 
    niels.gruettemeier@iosb-ina.fraunhofer.de, 
    felix.specht@iosb-ina.fraunhofer.de
}

\begin{document}
\maketitle

\begin{abstract}
  Cyber-physical production systems consist of highly specialized software and hardware components. Most components and communication protocols are not built according to the Secure by Design principle. Therefore, their resilience to cyberattacks is limited. This limitation can be overcome with common operational pictures generated by security monitoring solutions. These pictures provide information about communication relationships of both attacked and non-attacked devices, and serve as a decision-making basis for security officers in the event of cyberattacks. The objective of these decisions is to isolate a limited number of devices rather than shutting down the entire production system. In this work, we propose and evaluate a concept for finding the devices to isolate. Our approach is based on solving the \textsc{Critical Node Cut Problem with Vulnerable Vertices (CNP-V)}---an NP-hard computational problem originally motivated by isolating vulnerable people in case of a pandemic. To the best of our knowledge, this is the first work on applying CNP-V in context of cybersecurity.
\end{abstract}

\section{Introduction}
Cyber-physical systems combine physical processes with computation~\cite{7437398} and are networks of software and hardware components~\cite{lee2008cyber}. A subset of cyber-physical systems are cyber-physical production systems, they are in the focus of initiatives such as the Germany's platform Industrie 4.0 or the US Industrial Internet Consortium. Their main feature is flexibility to quickly adapt to new plant topologies~\cite{rehberger2016agent} and the integration of value-added services such as condition monitoring~\cite{specht_ae_2018}, or optimization~\cite{ottoAAAI,ottoTII}. This requires a high degree of interconnections between components~\cite{Wollschlaeger.2017}.

Cyber-physical production systems consist of highly specialized software and hardware components, such as programmable logic controller (PLCs), input or output devices (IO-Devices), and appliances. However, most components and communication protocols are not built according to the Secure by Design principle which limits their resilience to cyberattacks~\cite{specht_sdn_2022}. This limitation can be overcome with common operational pictures. They summarize the situation, update data on changing situations, exchanges data with internal and external systems, and collects information~\cite{kim2023study}. Security officers can then use these operational pictures to make decisions in the event of cyberattacks. Figure~\ref{fig_overview.pdf} illustrates the relation between security monitoring, common operational pictures and decisions. Security monitoring systems consist of three main components.

\vspace{-0.5em}
\bsfigurewidth[ht]{0.6}{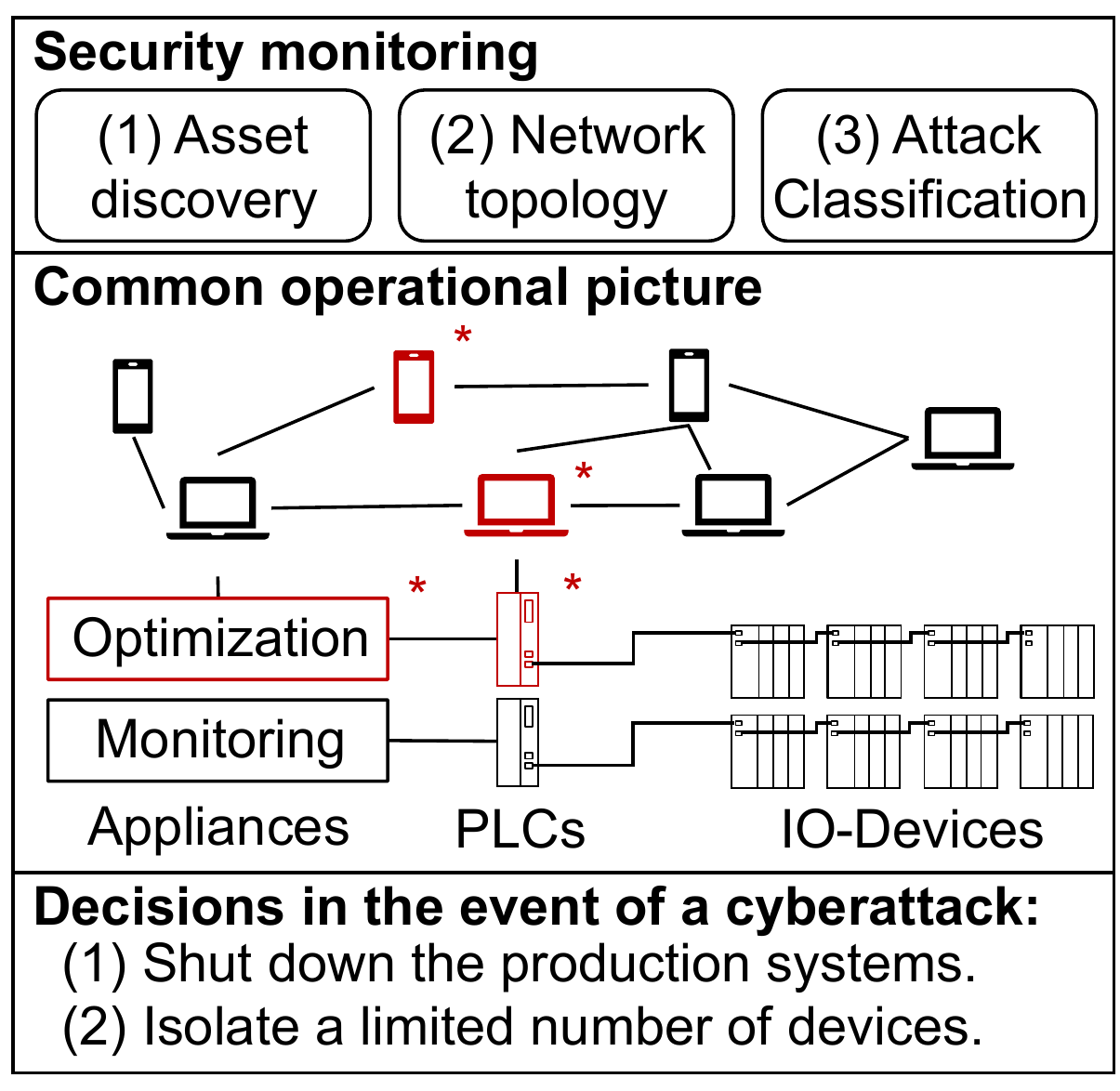}{Scope of this work: relation between security monitoring, common operational picture and decisions}

\vspace{-1em}
\textbf{Asset discovery:} Passive and active network analysis extract information about the used hardware and software components. Passive network analysis identifies basic component information, e.g. IP addresses and used protocols. Then active methods are used to gain specific component information, e.g. article number and device class.
Please note that only devices are considered in the following.

\medskip
\textbf{Network topology:} Passive network analysis is used to identify the communication relations between devices. Please note that specific communication protocols, such as SSH, OPC UA, Profinet or Ethercat, and their mapping to communication relations are not part of this work, only the relation knowledge is used.

\medskip
\textbf{Attack classification:} Cyberattacks are detected and classified according to the MITRE ATT\&CK framework~\cite{al2020learning,spechtINDIN2023}. Please note, that specific cyberattacks are not part of this work, only the knowledge of which devices are part of a of a cyberattack is used. The involved devices are marked with * (see.~Figure~\ref{fig_overview.pdf}) and named attacked devices.

The common operational picture consolidates information from asset discovery, network topology, and attack classification. It provides an overview of the devices, their communication relations and their classification as attacked or not-attacked. In the event of a cyberattack, a security officer can choose between two decisions:

\medskip
\textbf{Decision 1 (Worst-case):} The security officer shuts down the production system. It is the worst-case decision, because depending on the production process, it takes days or weeks to restart the production system.

\medskip
\textbf{Decision 2 (Best-case):} The security officer isolates only a limited number of devices. It is the best-case decision, because the production is not shut down. While the intuitive solution is to remove all attacked devices marked with a *, it is not always an optimal decision.

\medskip
In this work, a new computational problem, named \textsc{Security Node Problem with Vulnerable Vertices (SNP-V)}, is introduced to model this decision task. The solutions provided by SNP-V enable a security officer to isolate certain devices in order to reduce the impact of a cyberattack on a production system.

\section{Related Work}
This section describes the related work by initially introducing relevant approaches within the context of security applications, followed by the background of the \textsc{Critical Node Problem}.

\medskip
Security games have commonly been used for modeling security applications and determining optimal strategies for both defenders and adversaries~\cite{rosenfeld2017security, kamra2018policy}. A particular variant is network security games, which focus on the allocation of resources to network nodes or edges~\cite{wang2017non}. The concept of sharing defensive resources and reallocating resources on nodes in network security games has been introduced as a strategy to defend against contagious attacks~\cite{li2020defending, bai2021defending}. The application of deep learning to network security games has been increasingly proposed, aiming to discover optimal resource utilization strategies~\cite{xue2022nsgzero, li2023solving}. In network security, one commonly considered problem is the graph reachability reduction, where defenders aim to eliminate specific nodes or edges in order to prevent attacks~\cite{sheyner2002automated}. Zheng et al.~\citeyear{zheng2011active} introduce a solution based on binary edge classification for unknown edge costs. An associated approach involves utilizing the inverse geodesic length metric to identify and eliminate the most critical nodes or edges within a compromised network~\cite{najeebullah2018complexity, gaspers2019optimal}. Guo et al.~\citeyear{guo2022practical, guo2023scalable} investigate edge blocking in network graphs and propose different solution strategies based on both reinforcement learning and mixed-integer programming.

\medskip
The \textsc{Critical Node Problem (CNP)} is an NP-complete graph problem~\cite{A09}. It is studied extensively from an algorithmic point of view~\cite{SGL11,ASG13,HKKN16}. The problem finds application in several fields: In road networks it is used to plan emergency evacuations in disaster cases~\cite{VOTM11}; In context of cybersecurity it can be used to detect important vertices in a network that need to be protected~\cite{LTK18}. The \textsc{Critical Node Problem with Vulnerable Vertices (CNP-V)} is a generalization of the critical node problem and is originally motivated by isolating a vulnerable group of people in a social network in case of a pandemic~\cite{SGKS22}. So far, CNP-V has been studied from a purely theoretical point of view.

\section{Problem Definition}
This section describes the computational problem, named \textsc{Security Node Problem with Vulnerable Vertices (SNP-V)}. The problem has to be solved by an algorithm for automatic calculation of the devices to be isolated. SNP-V solutions enable a security officer to decide which devices need to be isolated to reduce the impact of a cyberattack on a production system.

Definition~\ref{def:graph} introduces the used graph notation that describes devices and connections. Definition~\ref{def:connections} describes direct and indirect connections between devices.

\medskip
\begin{definition}\label{def:graph}
    A \emph{graph} $G=(V,E)$ is a tuple consisting of a \emph{device set} $V$ and a \emph{connection set} $E$, where each~$e \in E$ is a two-element subset of~$V$. Given a subset~$C \subseteq V$ of devices that have to be removed, we let $G-C$ denote the graph obtained by deleting every device in~$C$ together with its incident connections.      
\end{definition}

\medskip
\begin{definition}[Connections] \label{def:connections}
    Let~$G=(V,E)$ be a graph, and let~$u \in V$ and~$v \in V$ be two devices. The devices~$u$ and~$v$ are called
    {\setlength{\leftmargini}{2em}
    \begin{enumerate}[label=(\roman*)]
        \item \emph{directly connected} if~$\{u,v\} \in E$,
        \item \emph{indirectly connected} if~$\{u,v\} \not \in E$, but there is a path from~$u$ to~$v$ in~$G$, and
        \item \emph{connected} if they are directly connected or indirectly connected.
    \end{enumerate}}
\end{definition}

\medskip
Next, assume that a subset of devices~$A \subseteq V$ is under attack, named \emph{attacked devices}. In this case, it is arguably problematic if an \emph{attacked device} is connected to any other device. This idea of potentially problematic connections is formalized in Definition~\ref{def:A-vul}: Intuitively, every device $v \in V$, which is directly or indirectly connected to an \emph{attacked device}~$a \in A$ forms a vulnerable connection with the attacked~device.

\medskip
\begin{definition}[$A$-vulnerability] \label{def:A-vul}
    Let $G=(V,E)$ be a graph and let $A \subseteq V$ be a subset of \emph{attacked devices}.
    A pair of devices~$u$ and~$v$ forms a \emph{vulnerable connection}~if
    \medskip
    {\setlength{\leftmargini}{2em}
    \begin{enumerate}[label=(\roman*)]
        \item $u$ and~$v$ are connected according to Definition~\ref{def:connections}, and
        \item $u \in A$ or~$v \in A$.
    \end{enumerate}}

    The \emph{$A$-vulnerability} of~$G$ is the number of device pairs forming a vulnerable connection in~$G$.
\end{definition}

\medskip
Definition~\ref{def:cnp-v} states the computational problem, named \textsc{Critical Node Problem with Vulnerable Vertices (CNP-V)}~\cite{SGKS22}. In CNP-V, one aims to decrease the number of \emph{vulnerable connections} according to Definition~\ref{def:A-vul}. The budget value $k \in \mathds{N}$ describes the number of removable devices from the graph in order to reduce the number of \emph{vulnerable connections}. The goal is, to remove these~$k$ devices in a way, such that the~\emph{$A$-vulnerability} of the resulting connection graph is minimal.

\medskip
\begin{definition}[CNP-V] \label{def:cnp-v}
    In the \textsc{Critical Node Problem with Vulnerable Vertices}, the input is a graph $G=(V,E)$, a set $A \subseteq V$, and an integer $k$. The task is to find a set~$C \subseteq V$ of size at most $k$, such that the~$A$-vulnerability of $G-C$, according to Definition~\ref{def:A-vul}, is minimal.
\end{definition}

\medskip
The set~$C$ from Definition~\ref{def:cnp-v} is called the \emph{solution} of a \textsc{CNP-V} instance~$(G,A,k)$. Even though it appears to be natural to exclusively remove \emph{attacked devices}, a \emph{solution} might in fact also contain non-attacked devices if a small budget~$k < |A|$ is given. Figure~\ref{fig_example.tex} illustrates a \emph{solution} containing a non-attacked device. The upper part shows a possible input graph for \textsc{CNP-V}. The red devices correspond to \emph{attacked devices~$A$}. The graph has 21 vulnerable connections: Three vulnerable connections with two \emph{attacked devices}, and~$3 \cdot 6 = 18$ connections with exactly one \emph{attacked device}. The lower part shows an optimal \emph{solution} for budget~$k=2$. Isolated devices are marked with~$\times$. After the isolation, the graph becomes disconnected and remains only three connected device pairs including an \emph{attacked device}.

\vspace{-0.5em}
  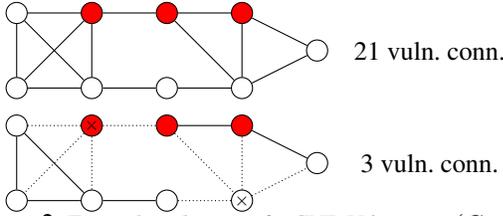
\begin{figure}[ht]
    \centering
		\resizebox{0.8\columnwidth}{!}{\begin{tikzpicture}
    \tikzstyle{knoten}=[circle,fill=white,draw=black,minimum size=8pt,inner sep=0pt]
    \tikzstyle{bez}=[inner sep=0pt]
    
    \node[knoten] (1)  at (0,1) {};
    \node[knoten] (2)  at (0,0) {};
    \node[knoten,fill=red] (3)  at (1,1) {};
    \node[knoten] (4)  at (1,0) {};
    
    \node[knoten,fill=red] (5)  at (2,1) {};
    \node[knoten] (6)  at (2,0) {};
    
    \node[knoten,fill=red] (7)  at (3,1) {};
    \node[knoten] (8)  at (3,0) {};
    \node[knoten] (9)  at (4,0.5) {};
    
    \draw[-]  (1) to (2);
    \draw[-]  (1) to (3);
    \draw[-]  (1) to (4);
    \draw[-]  (3) to (2);
    \draw[-]  (4) to (2);
    \draw[-]  (3) to (4);
    
    \draw[-]  (3) to (5);
    \draw[-]  (5) to (8);
    \draw[-]  (5) to (7);
    \draw[-]  (4) to (6);
    \draw[-]  (6) to (8);
    
    \draw[-]  (8) to (7);
    \draw[-]  (9) to (7);
    \draw[-]  (8) to (9);

    \node[bez] at (5.5,0.5) {21 vuln. conn.};
    
    \begin{scope}[yshift=-1.5cm]
    
    \node[knoten] (1)  at (0,1) {};
    \node[knoten] (2)  at (0,0) {};
    \node[knoten,fill=red] (3)  at (1,1) {\tiny{$\times$}};
    \node[knoten] (4)  at (1,0) {};
    
    \node[knoten,fill=red] (5)  at (2,1) {};
    \node[knoten] (6)  at (2,0) {};
    
    \node[knoten,fill=red] (7)  at (3,1) {};
    \node[knoten] (8)  at (3,0) {\tiny{$\times$}};
    \node[knoten] (9)  at (4,0.5) {};
    
    \draw[-]  (1) to (2);
    \draw[-,densely dotted]  (1) to (3);
    \draw[-]  (1) to (4);
    \draw[-,densely dotted]  (3) to (2);
    \draw[-]  (4) to (2);
    \draw[-,densely dotted]  (3) to (4);
    
    \draw[-,densely dotted]  (3) to (5);
    \draw[-,densely dotted]  (5) to (8);
    \draw[-]  (5) to (7);
    \draw[-]  (4) to (6);
    \draw[-,densely dotted]  (6) to (8);
    
    \draw[-,densely dotted]  (8) to (7);
    \draw[-]  (9) to (7);
    \draw[-,densely dotted]  (8) to (9);
    
    \node[bez] at (5.5,0.5) {3 vuln. conn.};
    \end{scope}
\end{tikzpicture}}
    \vglue 0ex plus 0.5ex minus 0.5ex
		\vspace{-1em}
    {\caption{\small Example \emph{solutions} of a CNP-V instance~$(G,A,k)$.}\label{fig_example.tex}}%
  \end{figure}

\vspace{-1.5em}
\subsection{Security Node Problem with Vulnerable Vertices}
This section introduces the \textsc{Security Node Problem with Vulnerable Vertices (SNP-V)}. It is a new variant of CNP-V by adding a secondary optimization goal to improve the stability of the remaining network after isolating devices. The motivation of the secondary optimization goal is the lack of uniqueness for the \emph{solution} of a CNP-V instance~$(G,A,k)$. More precisely, a lexicographic optimization is used. Among all solutions that minimize vulnerability, one aims to find the one that provides the best connectivity of the remaining graph.

\medskip
Figure~\ref{fig_no_unique.tex} illustrates an example with two possible \emph{solutions}~$C_1 = \{v_1\}$ and~$C_2=\{v_2\}$.
Even though, both \emph{solutions} provide the same~\emph{$A$-vulnerability}, isolating~$v_2$ also disconnects the non-attacked device~$v_3$ from the remaining graph. Thus, in security applications, deleting~$v_1$ is clearly a better choice. Motivated by this fact, the problem definition of CNP-V is adapted by adding a secondary goal focussing on the connectivity of non-attacked devices.

  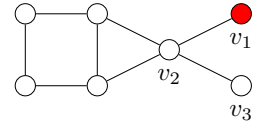
\begin{figure}[ht]
    \centering
		\resizebox{0.4\columnwidth}{!}{\begin{tikzpicture}
    \tikzstyle{knoten}=[circle,fill=white,draw=black,minimum size=8pt,inner sep=0pt]
    \tikzstyle{bez}=[inner sep=0pt]
    
    \node[knoten] (-1)  at (-1,1) {};
    \node[knoten] (0)  at (-1,0) {};
    \node[knoten] (1)  at (0,1) {};
    \node[knoten] (2)  at (0,0) {};
    \node[knoten,label=below:$v_2$] (3) at (1,0.5) {};
    \node[knoten,fill=red,label=below:$v_1$] (4)  at (2,1) {};
    \node[knoten,label=below:$v_3$] (5)  at (2,0) {};

    \draw[-]  (1) to (2);
    \draw[-]  (1) to (3);
    \draw[-]  (2) to (3);
    \draw[-]  (3) to (4);
    \draw[-]  (3) to (5);
    \draw[-]  (1) to (-1);
    \draw[-]  (0) to (2);
    \draw[-]  (-1) to (0);
\end{tikzpicture}}
    \vglue 0ex plus 0.5ex minus 0.5ex
		\vspace{-1em}
    {\caption{\small Example graph of an instance of CNP-V with~$k=1$, where the solution is not unique.}\label{fig_no_unique.tex}}%
  \end{figure}

This idea of healthy connections is formalized in Definition~\ref{def:healthy}: Intuitively, every device $v \in V$ and $v \notin A$, which is directly or indirectly connected to a non-\emph{attacked device} forms a healthy connection.

\medskip
\begin{definition}[$A$-healthiness]\label{def:healthy}
    Let~$G=(V,E)$ be a graph, and let~$A$ be a subset of \emph{attacked devices}.
    A pair of devices~$u$ and~$v$ forms a \emph{healthy connection}, if it is not a vulnerable connection. The~\emph{$A$-healthiness} of~$G$ is the number of device pairs forming a healthy connection in~$G$.
\end{definition}

\medskip
Next a lexicographic optimization is defined. The objective is to find a solution $C$ whose removal results in a graph~$G'=G-C$ which has maximum~\emph{$A$-healthiness} among all~$G'$ with minimum \emph{$A$-vulnerability}. Definition~\ref{def:targetFunc} defines this idea as a new objective function.

\medskip
\begin{definition}[Objective Value] \label{def:targetFunc}
Let~$G=(V,E)$ be a graph and let~$A \subseteq V$. Given a subset~$C \subseteq V$, let $G':=G-C$. Moreover, let
 {\setlength{\leftmargini}{2em}
 \begin{enumerate}[label=(\roman*)]
  \item $\text{\rm{vul}}(G')$ denote the~\emph{$A$-vulnerability} of~$G'$, and let
  \item $\text{\rm{heal}}(G')$ denote the~\emph{$A$-healthiness} of~$G'$.
 \end{enumerate}}
 The \emph{objective value}~$\Phi(G')$ is defined as:
  \begin{align*}
    \Phi(G') := (|V|^2 +1) \cdot \text{\rm{vul}}(G') - \text{\rm{heal}}(G') \in \mathds{Z}
 \end{align*} 
\end{definition}

\medskip
Note that a device set~$C$ minimizing~$\Phi(G')$ in fact provides a solution of our lexicographic optimization: Since~$|V|^2$ is an upper bound of the total number of device pairs, $\text{heal}(G') < |V|^2+1$. Consequently,~$(|V|^2 +1) \cdot \text{vul}(G')$ dominates the term in a way that minimizing~$\Phi(G')$ minimizes~$\text{vul}(G')$. Additionally, the larger the number of healthy connections, the smaller~$(-\text{heal}(G'))$. Therefore, we obtain a solution of our lexicographic optimization problem, when computing a \emph{solution}~$C$ minimizing~$\Phi(G')$.

\medskip
Definition~\ref{def:scnp-v} introduces the \textsc{Security Node Problem with Vulnerable Vertices (SNP-V)}.

\medskip
\begin{definition}[SNP-V] \label{def:scnp-v}
    In the \textsc{Security Node Problem with Vulnerable Vertices (SNP-V)}, the input is a graph $G=(V,E)$, a set $A \subseteq V$, and an integer $k$. The task is to find a \emph{solution}~$C \subseteq V$ of size at most $k$, such that~$\Phi(G')$, according to Definition~\ref{def:targetFunc}, is minimal.
\end{definition}

\section{Solving SNP-V}

\subsection{SNP-V ILP Definition}
This section introduces an Integer Linear Programming (ILP) formulation of SNP-V. This formulation allows the usage of standard solver implementations. It adapts a standard ILP formulation for the \textsc{Critical Node Problem}~\cite{A09}.

\medskip
\begin{definition}[SNP-V-ILP] \label{def:snp-ilp}
    Let~$(G=(V,E),A,k)$ be an instance of~SNP-V. The SNP-V-ILP formulation is defined as follows:
    
    \begin{subequations}
        \begin{align}
        \min    \quad & (|V|^2+1) \cdot \underbrace{ \sum_{\{i,j\} \cap A \neq \emptyset} \: u_{ij}}_{\emph{$A$-vulnerability}} - \underbrace{ \sum_{\{i,j\} \cap A = \emptyset} \: u_{ij} }_{\emph{$A$-healthiness}} \label{eq:def:snp-ilp_subeq1} \\
        \text{s.t.:} 	    \quad & \sum_{i \in V} v_i \leq k \label{eq:def:snp-ilp_subeq2} \\
                            \quad & (1-v_i) + (1-v_j) - 2 u_{ij} \geq 0, \: \forall \{i,j\} \in E \label{eq:def:snp-ilp_subeq3} \\ 
                            \quad & (1-v_i) + (1-v_j) - 2 u_{ij} \leq 1, \: \forall \{i,j\} \in E \label{eq:def:snp-ilp_subeq3b} \\ 
                            \quad & u_{ij} + u_{jk} - u_{ki} \leq 1, \forall (i,j,k) \in V^3 \label{eq:def:snp-ilp_subeq4a} \\
                            \quad & u_{ij} - u_{jk} + u_{ki} \leq 1, \forall (i,j,k) \in V^3 \label{eq:def:snp-ilp_subeq4b} \\
                            \quad & - u_{ij} + u_{jk} + u_{ki} \leq 1, \forall (i,j,k) \in V^3 \label{eq:def:snp-ilp_subeq4c} \\
                            \quad & u_{ij} \in \{0,1\}, \: \forall (i,j) \in V^2 \label{eq:def:snp-ilp_subeq5} \\
                            \quad & v_{i} \in \{0,1\}, \: \forall i \in V \label{eq:def:snp-ilp_subeq6}
        \end{align}
    \end{subequations}
\end{definition}

\medskip
The intuition behind the formulation provided in Definition~\ref{def:snp-ilp} is as follows: There are two types of variables. First, for every device~$i$, there is a variable~$v_i \in\{0,1\}$ corresponding to the choice of isolated devices in the following sense: the variable~$v_i$ is assigned with~$1$ in a solution of the ILP if and only if the corresponding device~$i$ is contained in the solution~$C$ for SNP-V. Second, for every pair~$(i,j) \in V^2$, there is a variable~$u_{ij} \in \{0,1\}$ corresponding to connectivity of devices in the following sense: the variable~$u_{ij}$ is assigned with~$1$ in a solution of the ILP if and only if the corresponding devices~$i$ and~$j$ are connected according to Definition~\ref{def:connections} in the solution graph~$G-C$ of SNP-V. Obviously, the Constraints~\eqref{eq:def:snp-ilp_subeq5} and \eqref{eq:def:snp-ilp_subeq6} define binary domains for the variables.

Constraint~\eqref{eq:def:snp-ilp_subeq2} ensures that the number of selected devices does not exceed the budget $k$.
Constraints~\eqref{eq:def:snp-ilp_subeq3} and~\eqref{eq:def:snp-ilp_subeq3b} ensure that two directly connected devices are connected in the solution if none of these devices is isolated.
The Constraints~\eqref{eq:def:snp-ilp_subeq4a}--\eqref{eq:def:snp-ilp_subeq4c} ensure that the connectivity is transitive, that is, for each~$(i,j,k) \in V^3$ we have~$u_{ik}=1$ if~$u_{ij}=u_{jk}=1$. Finally, the objective~\eqref{eq:def:snp-ilp_subeq1} models the objective function according to Definition~\ref{def:targetFunc}. 

\subsection{Solving SNP-V on Instances with a Small Budget}
This section states another approach for solving SNP-V. This approach is motivated by reducing the running time. Note that applying standard ILP solvers on the formulation given in Definition~\ref{def:snp-ilp} results in a superpolynomial-time algorithm for SNP-V. Since the \textsc{Critical Node Problem} is NP-hard~\cite{A09}, the more general SNP-V is NP-hard as well. Thus, for arbitrary instances, there is presumably no significant running time improvement over the ILP approach. However, in case of SNP-V, one might obtain a potential speedup for many real-world instances: First, it is a reasonable assumption that the deletion budget~$k$ is relatively small in comparison to the total number of devices. In fact, since an instance becomes trivial if we are allowed to delete all attacked devices, we may assume that the deletion budget is never larger than the number of attacked devices. Second, by the authors experience, many real-world networks have many devices with degree one. Herein, a device is a \emph{degree-one device} if it has exactly one incident connection. Therefore, it is highly recommended considering an algorithm that is fast if the budget~$k$ is relatively small and the number of degree-one devices is relatively large.

Schestag et al.~\citeyear{SGKS22} describe an algorithm for CNP-V with running time~$\mathcal{O}(|V|^k \cdot (|V|+|E|))$. Note that only the deletion budget~$k$ appears in the exponent in this running time. The algorithm is straight forward: Let~$(G=(V,E),A,k)$ be an input instance of CNP-V. Iterate over every possible device set~$C \subseteq V$ of size at most~$k$, and compute the~$A$-vulnerability~$s$ of~$G-C$. Finally, return the set~$C$, where~$s$ is minimal. We adapt this approach by computing the objective value from Definition~\ref{def:targetFunc} instead of the~$A$-vulnerability and by limiting the considered device sets~$C \subseteq V$ to sets that do not contain non-attacked degree-one devices. This adaption is described in two steps. First, we formally state and prove that non-attacked degree-one devices might in fact be excluded from the search. Second, we provide the adapted algorithm as pseudocode.

\begin{proposition} \label{Prop:NonAttDeg1}
    Let~$(G=(V,E),A,k)$ be an instance of \textsc{SNP-V}. Furthermore, let~$D \subseteq V \setminus A$ be the set of non-attacked degree-one devices. Then, there exists a solution~$C$ with~$C \cap D = \emptyset$.
\end{proposition}

\begin{proof}
Let~$C$ be a set of devices with~$C \cap D \neq \emptyset$. We describe how to transform~$C$ into a set~$C'$ with
\begin{enumerate}
    \item[$a)$] $|C'| \leq |C|$,
    \item[$b)$] $\Phi(G-C') \leq \Phi(G-C)$, and
    \item[$c)$] $|C' \cap D| < |C \cap D|$.
\end{enumerate}
Note that this implies the statement as repeatedly applying this transformation converts any solution into a solution not containing devices from~$D$. Since~$C \cap D \neq \emptyset$, there exists some~$v \in C \cap D$. If~$\Phi(G-(C\setminus\{v\})) \leq \Phi(G-C)$, then~$C':=C \setminus \{v\}$ clearly satisfies~$a)$,~$b)$, and~$c)$. Thus, we consider the case where~$\Phi(G-(C\setminus\{v\})) > \Phi(G-C)$. Since removing non-attacked devices never increases the healthiness, we conclude that~$G-(C \setminus \{v\})$ contains some attacked device forming a vulnerable connection with~$v$. We thus have~$|A \cap X| \geq 1$, where~$X$ denotes the connected component in~$G-(C\setminus\{v\})$ that contains~$v$. Since~$v$ has degree one, it has a unique neighbor~$w \in X$. We next show that deleting~$w$ instead of~$v$ provides a device set satisfying~$a)$,~$b)$, and~$c)$. We define~$C':= C \setminus \{v\} \cup \{w\}$. Note that~$|C'|=|C|$. It remains to prove~$b)$ and~$c)$. To this end, consider the following cases:

\medskip
\textbf{Case 1:} $w$ is non-attacked \textbf{.}
Then, since there exists some attacked device in~$X$, the device~$w$ has at least two neighbors. Therefore,~$|C' \cap D| < |C \cap D|$. Moreover, note that removing~$w$ also destroys all~$|A \cap X|$ vulnerable connections including~$w$. Thus, we have
\begin{align*}
    \text{vul}(G-C') \leq \text{vul}(G-C) - |A \cap X|.
\end{align*}
Since~$|A \cap X| \geq 1$, we have~$\Phi(G-C')<\Phi(G-C)$.

\medskip
\textbf{Case 2:} $w$ is attacked \textbf{.}
Then,~$w \not \in D$ and therefore~$|C' \cap D| < |C \cap D|$.
Moreover, note that removing~$w$ also destroys~$|X|-2$ vulnerable connections with the devices in~$X \setminus \{v,w\}$. Thus, we have
\begin{align*}
    \text{vul}(G-C') \leq \text{vul}(G-C) - (|X|-2).
\end{align*}
If~$|X| > 2$ we have~$\text{vul}(G-C') < \text{vul}(G-C)$ and therefore~$\Phi(G-C')<\Phi(G-C)$. If~$|X|=2$, then~$X=\{v,w\}$ and both graphs~$G-C$ and~$G-C'$ have the exact same vulnerable and healthy connections. Then,~$\Phi(G-C')=\Phi(G-C)$. Summarizing, it holds that~$\Phi(G-C')\leq \Phi(G-C)$.
\end{proof}

We next provide the algorithm as pseudocode. In Line~\ref{Line:FilterDegOne} of Algorithm~\ref{alg:CyberSeg},~$D$ is set to be the device set containing all non-attacked degree-one devices that may be excluded from the computation according to Proposition~\ref{Prop:NonAttDeg1}. In Line~\ref{Line:SubsetFamily},~$\mathcal{C}$ is set to be the family containing all subsets of~$V \setminus D$ with size at most~$k$ . Afterwards, in Lines~\ref{Line:CyberSeg-StartFor}--\ref{Line:CyberSeg-EndFor}, one computes the objective value~$s$ from Definition~\ref{def:targetFunc} for each~$G-C'$ with~$C' \in \mathcal{C}$ and stores the tuple~$(s,C')$. Finally, the subset~$C$ with minimum score is returned, cf. Line~\ref{Line:CyberSeg-Return}. The details of the computation of the score in Line~\ref{Line:CyberSeg-Score} are described in Algorithm~\ref{alg:CyberSeg-Score}.

\begin{algorithm}[h]
    \caption{CyberSeg-Direct}
    \label{alg:CyberSeg}
    \textbf{Input}: $G=(V,E)$, $A$, $k$\\
    \textbf{Output}: $C$
    \begin{algorithmic}[1]
    \STATE $D \gets $ degree-one devices in~$V \setminus A$ \label{Line:FilterDegOne}
    \STATE $\mathcal{C} \gets \text{combinations}(V \setminus D, k)$ \label{Line:SubsetFamily}
    \STATE $ F \gets \emptyset $
    \FORALL{$C' \in \mathcal{C} $} \label{Line:CyberSeg-StartFor}
        \STATE $G' \gets G - C'$
        \STATE $s \gets \text{score}(G', A) $ \label{Line:CyberSeg-Score}
        \STATE $F \gets F \cup \{(s, C')\}$  
    \ENDFOR \label{Line:CyberSeg-EndFor}
    \STATE \textbf{return} $\underset{C}{\arg\min}(\{ s \mid (s,C) \in F \})$ \label{Line:CyberSeg-Return}
\end{algorithmic}
\end{algorithm}

In Line~\ref{Line:DFS} of Algorithm~\ref{alg:CyberSeg-Score}, the connected components of the input graph are computed using a standard depth-first-search. In Lines~\ref{Line:Score-StartFor}--\ref{Line:Score-EndFor}, one iterates over all connected components. For each such component, one stores the number of vulnerable connections in a variable~$b$ using a formula by Schestag et al.~\citeyear{SGKS22} in Line~\ref{Line:VulInComp}. With~$b$ at hand, the~$A$-vulnerability (Line~\ref{Line:addToVul}) and the~$A$-healthiness (Line~\ref{Line:addToHeal}) of the current component are added to the total~$A$-vulnerability ($A$-healthiness) of~$G$. Finally, the score according to Definition~\ref{def:targetFunc} is returned, cf. Line~\ref{Line:scoreReturn}.

\begin{algorithm}[h]
    \caption{CyberSeg-Score}
    \label{alg:CyberSeg-Score}
    \textbf{Input}: $G=(V,E)$, $A$ \\
    \textbf{Output}: $s$
    \begin{algorithmic}[1]
    \STATE $\hat{C} \gets \text{components}(G)$ \label{Line:DFS}
    \STATE $g, h \gets 0$
    \FORALL{$(V',E') \in \hat{C} $} \label{Line:Score-StartFor}
    \STATE $a_c = 0$
    \FORALL{$v \in V'$}
        \IF{$v \in A$}
        \STATE $a_c \gets a_c + 1$
        \ENDIF
    \ENDFOR    
    \STATE $b \gets \binom{a_c}{2} + a_c \cdot (|V'|-a_c)$ \label{Line:VulInComp}
    \STATE $h \gets h + b $ \label{Line:addToVul}
    \STATE $ g \gets g + (\binom{|E'|}{2} -b) $ \label{Line:addToHeal}
    \ENDFOR \label{Line:Score-EndFor}
    \STATE \textbf{return} $(|V|^2+1) \cdot h - g $ \label{Line:scoreReturn}
\end{algorithmic}
\end{algorithm}

\subsection{Solving SNP-V with Greedy Algorithm}
Due to the NP-hardness of SNP-V, finding an exact solution becomes intractable for larger instances. This section provides a simple scheme for a greedy heuristic that provides a trade-off between running time and solution quality. The scheme is based on Algorithm~\ref{alg:CyberSeg}. Note that Algorithm~\ref{alg:CyberSeg} finds an optimal set of~$k$ devices to remove in running time~$\mathcal{O}(|V \setminus D|^k \cdot (|V|+|E|))$; a running time that crucially depends on the deletion budget~$k$. In the following, we use this algorithm to solve instances with some pre-defined smaller budget~$x<k$. Intuitively, the algorithm does not aim to find the best `large' solution of size~$k$, but repeatedly finds the best `small' solution of size~$x<k$. This is done as many times as~$x$ fits into~$k$. Algorithm~\ref{alg:CyberSeg-Greedy} describes this heuristic.
\begin{algorithm}[h]
    \caption{CyberSeg-Greedy}
    \label{alg:CyberSeg-Greedy}
    \textbf{Input}: $G=(V,E)$, $A$, $k$, $x$\\
    \textbf{Output}: $C$
    \begin{algorithmic}[1] 
    \STATE $C \gets \emptyset$ \label{Line:EmptySol}
    \WHILE{$x<k$} \label{Line:BeginWhileGreedy}
        \STATE $C' \gets $\texttt{CyberSeg(}$G-C,A,x$\texttt{)} 
        \STATE $C \gets C~\cup~C'$
        \STATE $k \gets k-x$
    \ENDWHILE \label{Line:EndWhileGreedy}
    \STATE $C \gets C~\cup~$\texttt{CyberSeg(}$G-C,A,k$\texttt{)} \label{Line:RemainingGreedy}
    \STATE \textbf{return} $C$
\end{algorithmic}
\end{algorithm}

In Line~\ref{Line:EmptySol}, the solution~$C$ is initialized as an empty set. In Lines~\ref{Line:BeginWhileGreedy}--\ref{Line:EndWhileGreedy}, the current solution is extended by deleting the current best device set of size~$x$ as long as this fits into the deletion budget. Finally, in Line~\ref{Line:RemainingGreedy}, the solution~$C$ is extended by removing the remaining~$k\leq x$ devices. Observe that the greedy algorithm has a running time of~$|V|^x \cdot (|V|+|E|)^{\mathcal{O}(1)}$. It depends on the chosen parameter~$x$. That is, the larger~$x$, the slower the algorithm, while---on the other hand---larger values of~$x$ are more likely to provide a better solution quality. The choice of~$x$ outlines a trade-off between solution quality and running time.

\section{First Experimental Results}
This section describes our first empirical results. Three different datasets are used: the Zachary's karate club data set~\cite{Z77}, a synthetic data set, and a real world data set describing a cyber-physical production system from the Smart Factory OWL (SFOWL). Zachary's karate club dataset describes a social network of a karate club. It describes conflict and fission in small groups. The data set is not related to the topic of security, but is often used to analyze this type of problem. In this work, individuals are devices or attacked devices and social relationships are connections. The synthetic data set is created by a graph generator~\cite{storer2001introduction}. The generator creates a full~$r$-ary tree of 50 devices and a branching factor of 5. The generated topology simulates a possible network of automation components.
This topology is used as it comes close to a cyber-physical production system as illustrated in Figure~\ref{fig_synthetic.pdf}. This example consists of three production modules. Each module has a PLC, IO-Devices and virtual appliances, such as condition monitoring or optimization as described above. A central PLC controls the overall process. The SFOWL dataset describes a real-world production system with 288 devices and 737 connections, illustrated in Figure~\ref{fig_SFOWL.pdf}.

\bsfigurewidth[ht]{0.9}{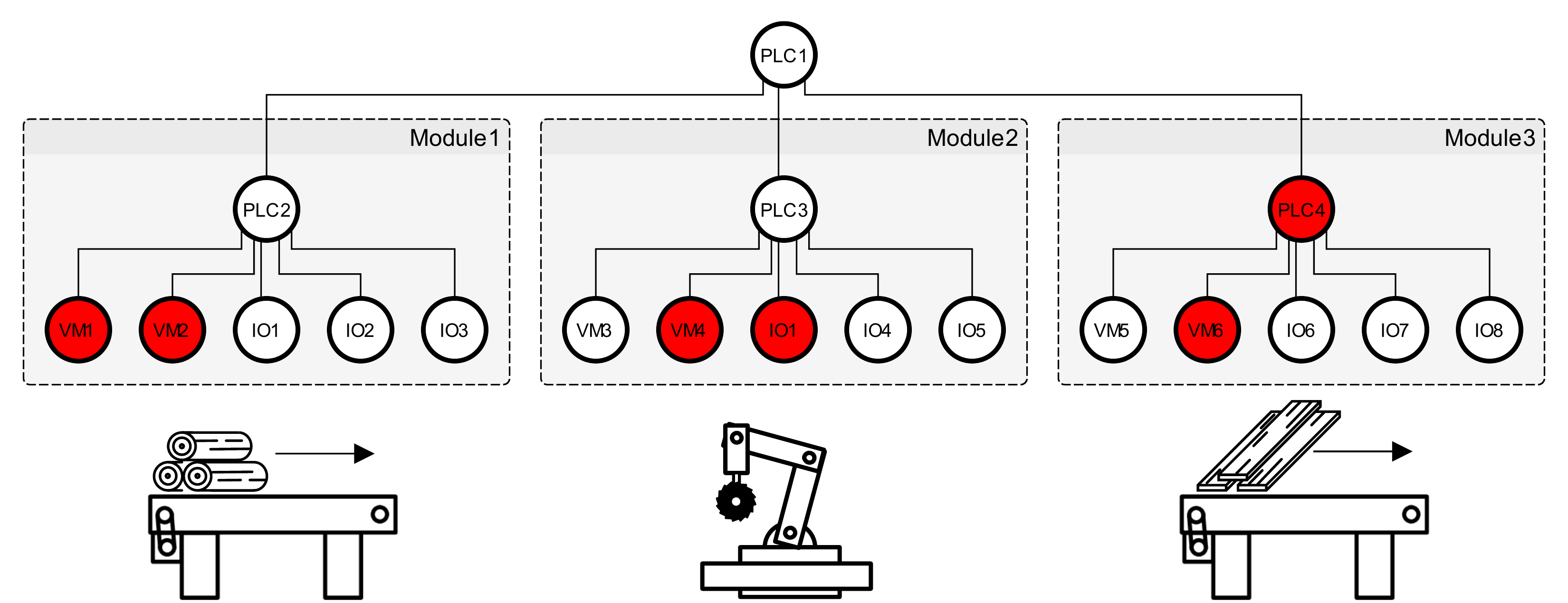}{Synthetic dataset example}
\bsfigurewidth[ht]{0.5}{fig_SFOWL.pdf}{SFOWL dataset}

\vspace{-1em}
Table~\ref{table_datasets} summarizes the metadata of the datasets. For each graph, three instances are generated by randomly choosing the set~$A$ of attacked devices. The value~$p \in \{ 0.1, 0.25, 0.5 \}$ corresponds to the fraction of attacked devices. More precisely, in a graph with~$n$ devices, we have~$\lceil p \cdot n \rceil$ attacked devices. These attacked devices are sampled uniformly.

The introduced SNP-V-ILP formulation (see Definition~\ref{def:snp-ilp}) is described by the Pyomo open-source optimization modeling language and the used solver is Gurobi. The CyberSeg implementations (see Algorithm~\ref{alg:CyberSeg} and Algorithm~\ref{alg:CyberSeg-Greedy}) are based on Python programming language. Each combination can be computed in parallel and each computation was performed on an Apple M2 Max with 12 cores. SNP-V-ILP and CyberSeg-Direct have a timeout limit set to 600 seconds, as we assume that this limit enables a timely response in a real-world scenario.

\begin{table}[!ht]
    \fontsize{8pt}{8pt}\selectfont
    \setlength{\tabcolsep}{3pt}
    \setlength{\extrarowheight}{1pt}
    \begin{center}
    \begin{tabular}{llccccc}
        \toprule
        \textbf{Dataset} & $p$ & \textbf{Devices} &  \textbf{d-1} & \textbf{Connections} & $v$ & $h$ \\
        \midrule
        \multirow{3}{*}{Karate}         & 0.1 & 34 & 1 & 78 & 96 & 465 \\
                                        & 0.25 & 34 & 1 & 78 & 236 & 325 \\
                                        & 0.5 & 34 & 1 & 78 & 425 & 136 \\\hline\hline
        \multirow{3}{*}{Synthetic}      & 0.1 & 50 & 38 & 49 & 235 & 990 \\
                                        & 0.25 & 50 & 31 & 49 & 522 & 703 \\
                                        & 0.5 & 50 & 19 & 49 & 925 & 300 \\\hline\hline
        \multirow{3}{*}{SFOWL}          & 0.1 & 288 & 121 & 737 & 7917 & 33411 \\
                                        & 0.25 & 288 & 100 & 737 & 18108 & 23220 \\
                                        & 0.5 & 288 & 66 & 737 & 31032 & 10296 \\
        \bottomrule
    \end{tabular}
    \end{center}
    \vspace{-1em}
    \caption[]{Dataset descriptions: d-1=degree-one devices in~$V \setminus A$, $v$=$A$-vulnerability, $h$=$A$-healthiness}%
    \setlength\belowcaptionskip{10em}
    \label{table_datasets}
    \end{table}
Table~\ref{table:v_h} summarizes a solutions quality analysis.
Note that CyberSeg-Direct and SNP-V-ILP both provide optimal solutions and therefore the same solution quality. As a consequence, the union of both results appears in the table as~\emph{Exact} and the individual timeouts are indicated.

\medskip
\textbf{SNP-V-ILP: }%
Independent of the budget~$k$ and the value of~$p$, solving SNP-V-ILP computes a solution within less than one minute when given the Karate and the Synthetic network. In contrast, for all instances on the real-world network SFOWL, no solution was found within the time limit of~600 seconds. The results show that in many instances it is not necessary to delete all attacked devices in order to remove all vulnerable connections. In detail, for~Karate with $p=0.5$, Synthetic with $p \in \{0.1, 0.25, 0.5\}$, it is sufficient to remove~$60\%$,~$80\%$,~$69\%$, and~$36\%$ of the total number of attacked devices, respectively. Note that the removed devices are not necessarily attacked devices, but devices that provide a high connectivity between attacked devices and the remaining graph. In the other two cases Karate with $p \in \{0.1, 0.25\}$, the results show that all attacked devices need to be removed to obtain $A$-vulnerability of~$v=0$. A further interesting observation can be made in case of~Synthetic with $p=0.5$, it is sufficient to remove nine devices to obtain $A$-vulnerability of~$v=0$, but removing ten devices also provides a better healthiness. Based on the solutions of the small networks, we find that the results indicate that solving SNP-V provides a good alternative to naively isolating all attacked devices, as it is sufficient to only shut down a few attacked devices. Thus, the computational problem SNP-V is a promising model for our practical use-case. However, it seems that the practical relevance of solving SNP-V via ILP is very limited as the computation takes too long. For the real-world network SFOWL, no solution was found even in cases with very small budgets~$k$. Summarizing, aiming for practical applications, it requires for more efficient algorithms and heuristics to solve SNP-V.

\medskip
\textbf{CyberSeg-Direct: }%
Since CyberSeg provides exact solutions like the SNP-V-ILP, the resulting $A$-vulnerability and $A$-healthiness are the same. Thus, we only discuss the running time of CyberSeg in this paragraph. For the solution quality, we refer to the previous paragraph discussing the SNP-V-ILP results. For the networks Karate and Synthetic, the algorithm was able to provide solutions for all~$k \in \{1,...,8\}$ within the time limit and ran into time-out for~$k \in \{9,10\}$. In case of~SFOWL, it computed solutions for~$k \in \{1,2,3\}$.
Recall that the theoretical running time of CyberSeg crucially depends on~$k$, which matches this experimental evaluation. While this algorithm exploits structural characteristics like degree-one-devices and small budgets, its practical relevance is limited as in the case of the SNP-V-ILP. However, in contrast to the SNP-V-ILP, the performance for~$k \leq 3$ can be seen as a promising subroutine that can be used as a plug-in for the greedy heuristic.

\medskip
\textbf{CyberSeg-Greedy: }%
Recall that CyberSeg-Direct is a subroutine in CyberSeg-Greedy which is evaluated for budgets~$x<k$. Motivated by the running times discussed in the previous paragraph, we performed experiments with~$x=3$ on all instances. On small instances with~$p=0.1$, the greedy algorithm always found an optimal solution. On small instances with~$p\in\{0.25,0.5\}$, the solution is close to the optimum. Observe that the budgets required to find a solution with~$A$-vulnerability~zero is at most the optimal budget plus two. In case of the SFOWL network with budget~$k=10$, the greedy algorithm found solutions decreasing the vulnerability by~$86\%$, $80\%$, and~$82\%$ in case of~$p=0.1$, $p=0.25$, and~$p=0.5$, respectively. In case of SFOWL with $p=0.1$ further experiments for larger values of~$k$ to determine for which budget a solution with vulnerability zero was found. With~$k=22$, the heuristic found a solution with vulnerability~$0$ and with healthiness~$3867$. That is, the greedy heuristic found a solution that removes~$75\%$ of the devices in comparison to the naive solution by simply removing all attacked devices. For the real-world data set, it is noticeable that the~$A$-healthiness values drop significantly compared to their initial ones. This is due to the fact that both~$A$-vulnerability and~$A$-healthiness decrease quadratically for larger~$k$ values and~$A$-healthiness is not taken into account by the algorithm as long as vulnerable nodes remain.

\begin{table}[!ht]
    \fontsize{8pt}{8pt}\selectfont
    \setlength{\tabcolsep}{1.3pt}
    \setlength{\extrarowheight}{1pt}
    \begin{center}  
        \begin{tabular}{c|c|cc|cc|cc|cc|cc|cc}
        \toprule
        \multicolumn{1}{c}{} & \multicolumn{1}{c}{} & \multicolumn{4}{c}{\textbf{Karate}} & \multicolumn{4}{c}{\textbf{Synthetic}} & \multicolumn{4}{c}{\textbf{SFOWL}} \\ 
        \multicolumn{1}{c}{} & \multicolumn{1}{c}{} & \multicolumn{2}{c}{Exact} & \multicolumn{2}{c}{Greedy} & \multicolumn{2}{c}{Exact} & \multicolumn{2}{c}{Greedy} & \multicolumn{2}{c}{Exact} & \multicolumn{2}{c}{Greedy} \\ \hline
        $k$ & $p$ & $v$ & $h$ & $v$ & $h$ & $v$ & $h$ & $v$ & $h$ & $v$ & $h$ & $v$ & $h$ \\ \hline
        1 & 0.1 & 63 & 465 & 63 & 465 & 14 & 341 & 14 & 341 & 6175$^\dag$ & 25203$^\dag$ & 6175 & 25203 \\
        2 & 0.1 & 31 & 465 & 31 & 465 & 9 & 331 & 9 & 331 & 4213$^\dag$ & 16294$^\dag$ & 4213 & 16294 \\
        3 & 0.1 & 0 & 465 & 0 & 465 & 4 & 321 & 4 & 321 & 3530$^\dag$ & 13865$^\dag$ & 3530 & 13865 \\\hline
        4 & 0.1 & 0 & 465 & 0 & 465 & 0 & 321 & 0 & 321 & - & - & 3159 & 13865 \\
        5 & 0.1 & 0 & 465 & 0 & 465 & 0 & 321 & 0 & 321 & - & - & 2924 & 13370 \\
        6 & 0.1 & 0 & 465 & 0 & 465 & 0 & 321 & 0 & 321 & - & - & 1808 & 8783 \\
        7 & 0.1 & 0 & 465 & 0 & 465 & 0 & 321 & 0 & 321 & - & - & 1519 & 8783 \\
        8 & 0.1 & 0 & 465 & 0 & 465 & 0 & 321 & 0 & 321 & - & - & 1172 & 5076 \\
        9 & 0.1 & 0$^\ddag$ & 465$^\ddag$ & 0 & 465 & 0 & 321 & 0 & 321 & - & - & 953 & 5076 \\
        10 & 0.1 & 0$^\ddag$ & 465$^\ddag$ & 0 & 465 & 0 & 321 & 0 & 321 & - & - & 821 & 4782 \\\hline\hline
        1 & 0.25 & 129 & 232 & 129 & 232 & 142 & 218 & 142 & 218 & 13294$^\dag$ & 16110$^\dag$ & 13294 & 16110 \\
        2 & 0.25 & 99 & 211 & 99 & 211 & 51 & 64 & 51 & 64 & 9919$^\dag$ & 10588$^\dag$ & 9919 & 10588 \\
        3 & 0.25 & 75 & 211 & 75 & 211 & 42 & 58 & 42 & 58 & 8614$^\dag$ & 8781$^\dag$ & 8614 & 8781 \\\hline
        4 & 0.25 & 39 & 44 & 52 & 211 & 25 & 51 & 33 & 52 & - & - & 7804 & 8131 \\
        5 & 0.25 & 26 & 19 & 30 & 211 & 19 & 51 & 24 & 46 & - & - & 7253 & 7629 \\
        6 & 0.25 & 9 & 232 & 9 & 211 & 14 & 51 & 19 & 46 & - & - & 5638 & 4953 \\
        7 & 0.25 & 5 & 232 & 5 & 211 & 9 & 51 & 14 & 46 & - & - & 4719 & 4468 \\
        8 & 0.25 & 0 & 325 & 1 & 210 & 2 & 1 & 9 & 46 & - & - & 3119 & 3129 \\
        9 & 0.25 & 0$^\ddag$ & 325$^\ddag$ & 0 & 210 & 0 & 45 & 4 & 46 & - & - & 2841 & 2974 \\
        10 & 0.25 & 0$^\ddag$ & 325$^\ddag$ & 0 & 210 & 0 & 75 & 0 & 40 & - & - & 2602 & 2898 \\\hline\hline
        1 & 0.5 & 282 & 79 & 282 & 79 & 257 & 98 & 257 & 98 & 22264$^\dag$ & 7140$^\dag$ & 22264 & 7140 \\
        2 & 0.5 & 207 & 79 & 207 & 79 & 92 & 23 & 92 & 23 & 15151$^\dag$ & 5356$^\dag$ & 15151 & 5356 \\
        3 & 0.5 & 144 & 56 & 144 & 56 & 77 & 23 & 77 & 23 & 12930$^\dag$ & 4465$^\dag$ & 12930 & 4465 \\\hline
        4 & 0.5 & 61 & 22 & 61 & 22 & 55 & 26 & 63 & 22 & - & - & 11930 & 4005 \\
        5 & 0.5 & 27 & 18 & 27 & 18 & 41 & 25 & 49 & 21 & - & - & 11054 & 3828 \\
        6 & 0.5 & 14 & 18 & 14 & 18 & 29 & 22 & 37 & 18 & - & - & 7889 & 2702 \\
        7 & 0.5 & 9 & 17 & 9 & 17 & 13 & 15 & 25 & 15 & - & - & 6902 & 2557 \\
        8 & 0.5 & 5 & 18 & 5 & 18 & 0 & 3 & 16 & 9 & - & - & 4462 & 1786 \\
        9 & 0.5 & 0$^\ddag$ & 17$^\ddag$ & 0 & 17 & 0$^\ddag$ & 21$^\ddag$ & 7 & 3 & - & - & 4088 & 1727 \\
        10 & 0.5 & 0$^\ddag$ & 18$^\ddag$ & 0 & 17 & 0$^\ddag$ & 55$^\ddag$ & 0 & 0 & - & - & 3729 & 1669 \\
        \bottomrule
        \end{tabular}
\end{center}
\vspace{-1em}
\caption[]{Solutions quality analysis: $v$=$A$-vulnerability, $h$=$A$-healthiness, $\dag$=ILP timeout, $\ddag$=CyberSeg-Direct timeout, $-$=timeout for both exact approaches}%
\label{table:v_h}
\end{table}

\section{Summary and Future Work}
This paper introduces CyberSeg, a novel approach to identify devices in cyber-physical production systems that should be isolated in the event of a cyber-attack. The \textsc{Security Node Problem with Vulnerable Vertices (SNP-V)} is presented for this. The necessary information, such as devices, attacked devices, and connection between devices, can be obtained from common security monitoring systems. SNP-V instances can be solved with CyberSeg and the results enable security officers to make decisions to reduce the impact of cyberattacks. CyberSeg helps to avoid worst-case decisions, such as shutting down the production system. Dependent on the production process, it takes days or weeks to restart the production system. We have shown that in all analyzed cases, the isolation of only a few devices leads to a healthy network. While this work provides a first step into studying SNP-V in context of security for cyber-physical production systems, there are several ways to extend this work. One idea is the formulation of SNP-V with specific domain knowledge to avoid isolation of devices critical for the production process. In cases where some devices are more expensive to be shut down than others, one may think of a problem version with weighted devices. Otherwise, in case of directed communication between devices one may need a problem version with directed graphs.

\bibliography{aaai24}

\end{document}